\documentclass{elsarticle}

\usepackage{natbib}
\usepackage[english]{babel}
\usepackage[utf8]{inputenc}
\usepackage{multirow}
\usepackage{amssymb,amsmath,amsthm,amsfonts}
\usepackage{epsfig}
\usepackage{color}
\usepackage{url}
\usepackage{enumerate, complexity}

\usepackage{graphicx}
\usepackage{xcolor}

\newcommand{\bb}[1]{\mathbb{#1}}
\newcommand{\ov}{\dot{1}}
\newcommand{\tc}[2]{\textcolor{#1}{#2}}

\newtheorem{theorem}{Theorem}
\newtheorem{proposition}[theorem]{Proposition}
\newtheorem{corollary}[theorem]{Corollary}
\newtheorem{question}[theorem]{Question}

\input amssym.def

\begin{document}

\begin{frontmatter}

\title{On Orthogonal Vector Edge Coloring}

\author[ufc]{Allen~Ibiapina MSc.}
\author[ufc]{Ana~Silva PhD\fnref{t2}}

\address[ufcmat]{ParGO, Departamento de Matem\'atica, Universidade Federal do Cear\'a, Fortaleza, Brazil}
\tnotetext[t1]{Funded by Conselho Nacional de Desenvolvimento Científico e Tecnológico - CNPq, grants: 304576/2017-4; 437841/2018-9; 425297/2016-0; 401519/2016-3.}

\fntext[emails]{anasilva@mat.ufc.br; allenr.roossim@gmail.com}



\begin{abstract}
Given a graph $G$ and a positive integer $d$, an orthogonal vector $d$-coloring of $G$ is an assignment $f$ of vectors of $\bb{R}^d$ to $V(G)$ in such a way that adjacent vertices receive orthogonal vectors. The orthogonal chromatic number of $G$, denoted by $\chi_v(G)$, is the minimum $d$ for which $G$ admits an orthogonal vector $d$-coloring. This notion has close ties with the notions of Lov\'asz Theta Function, quantum chromatic number, and many other problems, and  even though this and related metrics have been extensively studied over the years, we have found that there is a gap in the knowledge concerning the edge version of the problem. In this article, we discuss this version and its relation with other insteresting known facts, and pose a question about the orthogonal chromatic index of cubic graphs.
\end{abstract}

\end{frontmatter}


\section{Introduction}

We assume that the reader knows the basic terminology in Graph Theory and Linear Algebra; for an introduction on the subjects, we refer the reader to~\cite{Strang.book,W.book}. All the graphs are considered to be finite and simple, unless said otherwise.

A \emph{proper coloring} of a graph $G=(V,E)$ is a function $f \colon V \to \bb{N}$ such that $f(u)\neq f(v)$ whenever $uv \in E$. Because we only work with proper colorings in this paper, from now on we call it simply a \emph{coloring}; also, if $f$ uses $k$ colors, we refer to it as a \emph{$k$-coloring}. The \emph{chromatic number} of $G$, denoted by $\chi(G)$, is the minimum integer $k$ for which $G$ admits a $k$-coloring; and if $\chi(G)=k$, we say that $G$ is \emph{$k$-chromatic}. Deciding whether $\chi(G)\le k$ is $\NP$-complete even when $k=3$ and $G$ is a triangle-free graph with maximum degree at most~4~\cite{MP.99}.

An \emph{orthogonal vector $d$-coloring of $G$} (from now on called simply orthogonal $d$-coloring) is an assignment $f$ of vectors of $\bb{R}^d$ to $V(G)$ in such a way that adjacent vertices receive orthogonal vectors. This was originally defined in terms of the non-edges (i.e., non-adjacent vertices receive orthogonal vectors), and it was introduced to approach the Shannon capacity function, a problem arising in information theory~\cite{L.79}. Since then, it has been extensively studied~(see \cite{L.book}, Chapter~9 for an overview), as well as many other notions stemming from it, as for example the quantum chromatic number~\cite{GW.02,AHKS.2006,CMNS.06}.

One of the many investigated aspects when dealing with these representations is the minimum dimension for which it exists. Here, we call this parameter the \emph{orthogonal number of $G$} and denote it by $\pi(G)$. Observe that, given a $d$-coloring  $f$ of $G$, one can obtain an orthogonal $d$-coloring from $f$ by associating each color with a vector in the canonical base of $\mathbb{R}^{d}$. Therefore, 
\begin{equation}\pi(G)\leq \chi(G).\label{eq:pileqchi}\end{equation}

Little is known about the orthogonal number of a graph, and some stronger constraints have been imposed on the definition, as for instance, the \emph{faithful} orthogonal $d$-colorings~\cite{CSW.11,CELP.12}, and the \emph{general position} orthogonal $d$-coloring~\cite{LSS.89.00}. 

A natural and interesting question about the dimension is whether the inequality above is tight. Examples where the inequality is strict are not easy to find, and up to our knowledge, all the known examples are related to quantum mechanics and \emph{pseudo-telepathy games}. In~\cite{AHKS.2006}, they prove that most Hadamard graphs have this property, in~\cite{CMNS.06} they give an explicit example in~18 vertices, and in~\cite{SS.12} they give a more general result that produce smaller graphs than the Hadamard graphs (the smaller Hadamard graph $G$ having $\pi(G)<\chi(G)$ has 1609 vertices). The latter article relates the concepts of orthogonal coloring and quantum coloring to Kochen-Specker sets, a notion used initially to prove the inadequacy of using determinism to model quantum mechanics~\cite{KS.1967}. 

A \emph{Kochen-Specker set} is a set $S\subseteq \bb{R}^d$ for which there is no function $f:\bb{R}^d\rightarrow \{0,1\}$ such that, for every orthonormal bases $B\subseteq S$, we have $\sum_{u\in B}f(u)=1$. We mention that these sets are actually defined for general Hilbert spaces (we refer the reader to~\cite{HPSW.10} for a fine explanation), but that the restriction made will be enough for our purposes. Also, the proof given below has appeared before in more general contexts, but we reproduce it here for the sake of clarity. 




Given any set $S\subseteq \bb{R}^{d}$, the \emph{orthogonality graph of S}, denoted by $K(S)$, has $S$ as vertex set and is such that $u$ and $v$ are adjacent if and only if $u$ is orthogonal to $v$. 

\begin{theorem}
Let $S\subseteq \bb{R}^{d}$ be a Kochen-Specker set containing an orthonormal base of $\bb{R}^d$. Then, \[\chi(K(S))>\pi(K(S)) = d.\]
\end{theorem}
\begin{proof}
Let $G=K(S)$. First, observe that any orthogonal base of $\bb{R}^d$ is a clique in $G$, and must receive orthogonal vectors in any orthogonal coloring of $G$; hence, because $S$ contains an orthonormal base of $\bb{R}^d$, we get $\pi(G)\ge d$. 

Now, suppose by contradiction that $G$ has a $d$-coloring $c$, and let $f:\bb{R}^d\rightarrow \{0,1\}$ assign value $1$ to every $u$ in color class 1, and 0 to the remaining vectors. Consider an orthonormal base $B$ of $\bb{R}^{d}$ contained in $S$. Again, we know that $B$ is a clique of size $d$ in $G$, and therefore must contain exactly one vertex of each color class. This implies that $\sum_{u\in B}f(u) = 1$, contradicting the definition of a Kochen-Specker set.
\end{proof}


As already mentioned, finding a Kochen-Specker set is no trivial task, the first being presented as a proof of the famous Kochen-Specker Theorem~\cite{KS.1967} and containing 117 vectors. Now, it is known that any such set in $\bb{R}^3$ and $\bb{R}^4$ have at least~18 vertices~\cite{CEG.96,PMMM.05_2,PMMM.05}. It is also interesting to point out that in~\cite{HPSW.10} the authors show a subset $S\subseteq \bb{R}^4$ of the Kochen-Specker set presented in~\cite{P.91} for which $|S|=17$ and $4=\pi(K(S))<\chi(K(S))=5$. Because, as mentioned before, every Kochen-Specker set in dimension~4 must have at least 18~vertices, we know that $S$ is not a Kochen-Specker set, thus showing that being a Kochen-Specker set is not a necessary condition to have $\pi(K(S))<\chi(K(S))$.

Now, given that finite Kochen-Specker sets are known, we get that a (simple) graph $G$ with  $\chi(G)>\pi(G)$ is also known. From this graph, it is possible to construct graphs where the distance between $\chi(G)$ and $\pi(G)$ is as large as possible; it suffices to iteratively apply the join operation on two copies of $G$. Given graphs $G_1,G_2$, the \emph{join} of $G_1$ and $G_2$ is the graph obtained from $G_1\cup G_2$ by adding all possible edges between $V(G_1)$ and $V(G_2)$; it is denoted by $G_1\wedge G_2$. 

\begin{proposition}
Let $G_1, G_2$ be graphs. Then, \[\pi(G_1\wedge G_2) = \pi(G_1) + \pi(G_2).\]
\end{proposition}
\begin{proof}
For each $i\in\{1,2\}$, denote $V(G_i)$ by $V_i$, and let $f_{i} \colon V(G_{i}) \to \bb{R}^{d_{i}}$ be an orthogonal $d_{i}$-coloring of $G_{i}$. Also, let $G = G_{1}\wedge G_{2}$, and define $f \colon V(G)\to \bb{R}^{d_{1}+d_{2}}$ as follows. Set $f(v)=(f_{1}(v),0)$ if $v \in V_1$ and $f(v)=(0,f_{2}(v))$ if $v \in V_2$. One can verify that $f$ is an orthogonal $(d_{1}+d_{2})$-coloring of $G$; hence $\pi(G) \leq \pi(G_1) + \pi(G_2)$. 

Now let $f \colon V(G)\to \bb{R}^{d}$ be an orthogonal $d$-coloring of $G$. For every $u \in V_1$ and $v\in V_2$, we have that $f(u)$ is orthogonal to $f(v)$. In particular $f(v) \notin \operatorname{span}(f(V_1))$ and then $\operatorname{span}(f(V_1))^{\perp}\supseteq \operatorname{span}(f(V_2))$. Thus the dimension of $\operatorname{span}(f(V_1))\cup \operatorname{span}f(V_2)$ is equal to the dimension of $\operatorname{span}(f(V(G))$, which is equal to $d$. Because $f$ restriced to $V_i$ is an orthogonal coloring of $G_i$, we have that the dimension of $\operatorname{span}(f(V_i))$ is at least $\pi(G_{i})$ for each $i \in \{1,2\}$, and it follows that $d\geq \pi(G_{1})+\pi(G_{2})$. 
\end{proof}

The analogous equality applies to $\chi(G_1\wedge G_2)$ and is a folklore, easy to check, result. Now, if $G$ is such that $\chi(G)>\pi(G)$, and $H=G\wedge G$, then $\pi(H) = 2\pi(G)$ and $\chi(H) = 2\chi(G)$, i.e, $ \chi(H)- \pi(H) = 2(\chi(G)-\pi(G))$. Therefore, everytime we take the join, we double the gap between the chromatic number and the orthogonal number. 

\begin{corollary}\label{cor:biggap}
For every positive integer $k$, there exists $G$ such that $$\chi(G)-\pi(G)\ge 2^k.$$
\end{corollary}

Although, as we have seen, these aspects have been extensively investigated, we were unable to find any results on the related edge version of the problem (which is equal to investigate the parameter restricted to line graphs), and this is our main interest in this article. An \emph{orthogonal $d$-edge-coloring} of a graph $G$ is a function $f\colon E(G) \to \mathbb{R}^{d}$ such that $f(e)$ is orthogonal to $f(e')$ whenever $e$ and $e'$ are adjacent (i.e., $e\cap e'\neq \emptyset$). The minimum value $d$ such that there exists an orthogonal $d$-edge-coloring of $G$ is the \emph{orthogonal index of G} and is denoted by $\pi'(G)$. A \emph{$k$-edge-coloring}, the \emph{chromatic index} $\chi'(G)$ and a \emph{$k$-edge-chromatic graph} are defined similarly on the edges, and by arguments similar to the ones applied to vertices, we get that $$\pi'(G)\leq \chi'(G).$$

Observe that the edges incident to some vertex $v$ must be colored with $d(v)$ orthogonal vectors, and so $\Delta(G)\leq \pi'(G)$. Now, using Vizing's Theorem~\cite{V.64} we obtain the following inequalities:

$$\Delta(G)\leq \pi'(G)\leq \chi'(G) \leq \Delta(G)+1.$$

This tells us that the gap between $\pi(G)$ and $\chi(G)$ is at most~1 when $G$ is a line graph, i.e., Corollary~\ref{cor:biggap} does not hold for line graphs. Nevertheless, we can again ask if there exists a graph $G$ such that $\pi'(G) < \chi'(G)$. The answer remains positive and our example is the the Kochen-Specker set given in \cite{CEG.96}. Essentially, we show that the given set $S$ can also be assigned to edges of a graph $G(S)$ on 9 vertices. Each vertex of our graph represents a distinct orthogonal base $B_i$ of $\bb{R}^4$, and two vertices are adjacent if they share a common vector. We present the bases in the table below (the symbol $\ov$ denotes $-1$), and we use the same color to represent the same vector. Observe that each vector appears exactly twice. The graph itself is drawn in Figure~\ref{fig:KSgraph}.

\begin{table}[htb]
\begin{center}
\begin{tabular}{|c|c|c|c|c|c|c|c|c|}
\hline
$\mathbf{B_1}$ & $\mathbf{B_2}$ & $\mathbf{B_3}$ & $\mathbf{B_4}$ & $\mathbf{B_5}$ & $\mathbf{B_6}$ & $\mathbf{B_7}$ & $\mathbf{B_8}$ & $\mathbf{B_9}$\\
\hline
\tc{green}{0001} & \tc{green}{0001} & \tc{red}{$1\ov 1\ov$} & \tc{red}{$1\ov 1\ov$} & \tc{brown}{0010} & \tc{orange}{$1\ov\ov 1$} & \tc{blue}{$11\ov1$} & \tc{blue}{$11\ov1$} & \tc{purple}{$111\ov$}\\\hline

\tc{brown}{0010} & \tc{magenta}{0100} & \tc{orange}{$1\ov\ov 1$} & \tc{olive}{1111} & \tc{magenta}{0100} & \tc{olive}{1111} & \tc{purple}{$111\ov$} & \tc{yellow}{$\ov111$} & \tc{yellow}{$\ov111$} \\\hline

\tc{cyan}{1100} & \tc{teal}{1010} & \tc{cyan}{1100} & \tc{gray}{$10\ov0$} & \tc{darkgray}{1001} & \tc{pink}{$100\ov$}& \tc{red!50}{$1\ov00$} & \tc{teal}{1010} & \tc{darkgray}{1001} \\\hline

\tc{red!50}{$1\ov00$} & \tc{gray}{$10\ov0$} & \tc{blue!50}{0011} & \tc{green!50}{$010\ov$} & $01\ov0$ & \tc{pink}{$100\ov$} & \tc{blue!50}{0011}& \tc{green!50}{$010\ov$} & $01\ov0$ \\\hline
\end{tabular}
\caption{Kochen-Specker set $S$ on $\bb{R}^4$ used to construct $G(S)$.}\label{table:bases}
\end{center}
\end{table}

\begin{figure}[h!]
\centering \includegraphics[scale=0.5]{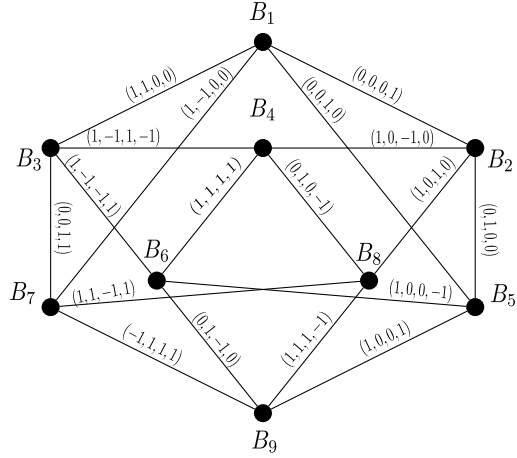}
\caption{Graph $G(S)$ related to the bases presented in Table~\ref{table:bases}. Each base is a vertex of $G(S)$, and two bases are adjacent if and only if they have a vector in common.}
\label{fig:KSgraph}
\end{figure}

%
%
%
%
%

\begin{theorem}
There exists a graph $G$ such that \[4=\pi'(G) < \chi'(G)=5.\]
\end{theorem}
\begin{proof}

Let $G$ be the graph on vertices $\{B_1,\ldots,B_9\}$ given above and edge set $\{B_iB_j\mid B_i\cap B_j\neq \emptyset\}$. First, note that $G$ is 4-regular, since each base has exactly four vectors and each vector appears in exactly two bases; hence we get the upper and lower bounds. Furthemore, the vectors themselves give us an orthogonal 4-edge-coloring, since the bases related to the vertices are orthogonal (i.e., vectors incident to the same vertex are orthogonal). It remains to show that $E(G)$ cannot be properly colored with~4 colors. 
Suppose otherwise and let $f$ be a 4-edge-coloring of $G$. Since the graph is 4-regular, we get that each vertex $v \in V(G)$ is incident to every color. In particular $f^{-1}(1)$ would be a perfect matching, a contradiction since $|V(G)|$ is an odd number. 
\end{proof}

The theorem above of course also tells us of the existence of a graph $H$ for which $4=\pi(H) < \chi(G) = 5$ (consider $H=L(G)$). A good question is whether we can have strict inequalities for graphs with orthogonal number smaller than~3. The answer is positive, as we state in the theorem below.
\begin{theorem}\label{theo:nogap}
If $G$ is a bipartite graph, then $\pi(G) = \chi(G)$; and if $\pi(G)=2$ then $G$ is bipartite. Furthermore, there exists a graph $G$ such that \[3 = \pi(G) < \chi(G)=4\].
\end{theorem}
\begin{proof}
For the first part, note that $\pi(G)\ge 2$ when $E(G)\neq \emptyset$, which implies $\pi(G)=\chi(G)$ when $G$ is bipartite.  Also, observe that for every $u\in \bb{R}^2$, up to reflection, there exists a unique unit vector $v\in \bb{R}^2$ orthogonal to $u$, and this implies that an odd cycle cannot admit an orthogonal 2-coloring. 

Now, for the second part, let $\mathbb{S}^{2}\subset \mathbb{R}^{3}$ be the unit sphere and $H=K(\mathbb{S}^{2})$ be its orthogonality graph. By construction, we know that $\pi(H)=3$, and in~\cite{GS.88} it has been proven that $\chi(H)=4$. Applying the De Bruijn-Erd\H{o}s Theorem~\cite{E.50}, that states that $\chi(X) = \max \{\chi(Y)\mid~ Y\subset X \text{ is finite}\}$ for every infinite graph $X$, we get that a graph with $3=\pi(G)<\chi(G)=4$ exists. 
\end{proof}

We then know that the gap occurs starting at graphs with orthogonal number equal to~3. Now, we can ask whether the same holds for the orthogonal index. Again because every $u\in \bb{R}^2$ is orthogonal to a unique line in $\bb{R}^2$, we get that if $G$ has an orthogonal 2-edge-coloring, then $G$ must be a collection of paths and even cycles, in which case $G$ has also a 2-edge-coloring. This means that the gap again cannot occur when $\pi'(G)\le 2$. As we have seen, the proof of the second part of the theorem above is non-constructive, hence it is not known whether such a graph can be a line graph. 
We then pose the following question.

\begin{question}\label{question}
Is there a graph $G$ such that $3= \pi'(G) < \chi'(G)=4$?
\end{question}

Clearly, such a graph would also give an explicit example in the proof of the second part of Theorem~\ref{theo:nogap}. Therefore, if it exists, finding such a graph should not be easy. We also comment on some ties with the Four Color Theorem which suggest that finding such a graph is indeed a hard task. 

Observe that, if $G$ is a 2-connected planar 3-regular graph, then by the Four Color Theorem and Tait's Theorem, we get that the answer to the question above is negative considering $G$, i.e., every 2-connected planar 3-regular graph is 3-edge-colorable and hence cannot give a positive answer to the question. Tait's Theorem says that the Four Color Theorem is equivalent to stating that every 3-regular 2-connected planar graph is 3-edge-colorable. Tait initially believed that every 3-regular 2-connected planar graph had to be hamiltonian, and if true, applying his theorem would give a beautiful proof for the Four Color Theorem. Indeed, if $G$ is a 3-regular hamiltonian graph, then we must have $\chi'(G)=3$: it suffices to remove a perfect matching of $G$, and coloring the hamiltonian cycle with two colors (the number of vertices in the graph must be even since $G$ is 3-regular and the sum of the degrees is even). Even though the gap on his proof was found around the time of announcement in 1878, it was not until 1946 that an explicit example was found by Tutte~\cite{T.46}. The presented graph, named after Tutte, has 46 vertices and disproved Tait's Conjecture about all such graphs being hamiltonian. Since then, numerous examples have been found, many with the help of Grinberg's Theorem~\cite{G.68}, and currently it is known that the smallest such graph has~38 vertices~\cite{HM.86}. 

Even though we know that $\pi'(G) = 3$ when $G$ is a 2-edge connected 3-regular planar graph, because of the Four Colour Theorem and Tait's Theorem as mentioned earlier, an interesting question is whether there exists a purely algebraic proof for this. Observe that, if so, then a negative answer to our question would give an alternative proof for the Four Color Theorem.

Besides the broad research around non-hamiltonian 3-regular 2-con\-nected graphs, a lot of work has been done on 3-regular 4-edge-chromatic graphs, which are also classified as \emph{Class 2} graphs according to Vizing's Theorem. Observe that these graphs are candidates for a positive answer to our question. Unfortunately, though, these graphs are also very hard to find, so much so that the minimal such graphs have been baptized \emph{snarks} by Martin Gardner~\cite{G.76}, evoking a Lewis Carroll's poem's rare creature. The Petersen graph was the first found snark, discovered in 1898, and up until 1975, when Isaacs constructed two infinite families of these graphs~\cite{I.75}, apart from the Petersen graph, only three other snarks were known~\cite{B.46,Tutte.Descartes.48,S.73}. Nowadays, these elusive creatures are the object of much reasearch and are known to play a central role in a number of conjectures~\cite{CMRS.98}. 

Finally, recall that, as seen in the proof of Theorem~\ref{theo:nogap}, the unit sphere $\mathbb{S}^{2}\subset \mathbb{R}^{3}$  is such that $3=\pi(H)<\chi(H)=4$, where $H=K(\mathbb{S}^{2})$. Observe that if $f$ is an orthorgonal edge-coloring of $G$, and $S = f(E(G))$, then there is a homomorphism from $L(G)$ to $K(S)$; also, we know that $S$ can be considered to be on $\mathbb{S}^2$, since it suffices to normalize the vectors in $S$. In~\cite{GS.88}, the authors also prove that the orthogonal graph of the rational sphere $\mathbb{Q}^2\subset\mathbb{R}^3$ is 3-chromatic. Therefore, if $G$ is a graph that answers Question~\ref{question} positively, then $L(G)$ is 4-chromatic and therefore cannot have a homomorphism to $K(\mathbb{Q}^2)$. Because a computer can only find representations rational coordinates, we get that no positive answer can be found with computer aid. We emphasize once more that no explicit finite graph $G$ with $3=\pi(G)<\chi(G)=4$ is known, and that its existence is only ensured by the De Brujin-Erd\H{o}s Theorem~\cite{E.50}. Below, we summarize the known discussed necessary conditions when restricted to cubic graphs.

\begin{proposition}
Let $G$ be a cubic 2-connected graph such that $3=\pi'(G)<\chi'(G)=4$. Then $G$ is non-hamiltonian and non-planar, and every orthogonal 3-coloring of $E(G)$ must use vectors with irrational components.
\end{proposition}

\bibliographystyle{plain}

\bibliography{biblio}


%
%
%

\end{document}